\documentclass[12pt,article]{amsart}
\usepackage{mathrsfs}
\usepackage{amssymb}
\usepackage{amsfonts}
\usepackage{amsbsy}
\usepackage{latexsym}
\usepackage{amssymb,latexsym,amsmath,amsthm}
\usepackage{framed}
\usepackage[colorlinks,linkcolor=blue]{hyperref}
\usepackage{graphicx}
\usepackage{xcolor}
\usepackage{epstopdf}
\usepackage{bm}
\theoremstyle{plain}

 \theoremstyle{remark} 

\newtheorem {theo} {\bf Theorem} [section]
\newtheorem {prop} [theo] {\bf Proposition}

\newtheorem {lem} [theo] {\bf Lemma}

\newtheorem {defi} {\bf Definition}[section]

\newtheorem{rem}{\bf Remark}[section]

\numberwithin{equation}{section}
\begin{document}
\title[Single-angle Radon samples based  reconstruction]{Single-angle Radon samples based  reconstruction  of   functions in refinable shift-invariant space}
\author{Youfa Li}
\address{College of Mathematics and Information Science\\
Guangxi University,  Nanning, China }
\email{youfalee@hotmail.com}
\author{Shengli Fan}
\address{CREOL College of Optics \& Photonics\\
University of Central Florida,  Orlando, FL 32816}
\email{shengli.fan@knights.ucf.edu}
\author{Yanfeng Huang}
\address{College of Mathematics and Information Science\\
Guangxi University,  Nanning, China }
\email{hyfbqy@163.com}
\thanks{Youfa Li is partially supported by Natural Science Foundation of China (Nos: 61961003, 61561006, 11501132),  Natural Science Foundation of Guangxi (Nos: 2018JJA110110, 2016GXNSFAA380049) and  the talent project of  Education Department of Guangxi Government  for Young-Middle-Aged backbone teachers.
}
\keywords{computerized tomography, single-angle Radon sample, refinable shift-invariant spaces, box-spline, positive definite function}
\subjclass[2010]{Primary 42C40; 65T60; 94A20}

\date{\today}

\begin{abstract}
The traditional approaches to computerized tomography (CT) depend on the samples of  Radon transform at multiple angles.
In optics, the real time imaging requires the reconstruction   of an object by the samples of Radon transform at a    \emph{single angle} (SA).
Driven by this and motivated by the connection between Bin Han's construction of wavelet frames (e.g \cite{Han1}) and Radon transform,  in refinable
shift-invariant spaces (SISs)  we investigate the SA-Radon sample based reconstruction problem. We have two main theorems. The  fist main theorem states   that,  any compactly supported  function in a SIS generated by   a  general refinable function can be determined by
its Radon  samples at an appropriate  angle. Motivated by the  extensive application of positive definite (PD) functions to
interpolation of scattered data, we also investigate the SA reconstruction problem in a class of (refinable)  box-spline generated SISs.
Thanks to the PD property of the Radon transform of such spline, our second main theorem states  that, the   reconstruction of  compactly supported functions in these spline generated  SISs  can be achieved  by the samples of  Radon transform  at almost every   angle.
Numerical simulation is conducted to check the result.
\end{abstract}
\maketitle

\section{Introduction}
\subsection{CT and Radon transform}
We start with the X-ray computerized tomography (CT) on $\mathbb{R}^{2}$. 
Its core mathematics   includes  the Radon transform and its inversion.
For a function $f\in L^{2}(\mathbb{R}^{2})$,
the  Radon transform, w.r.t a fixed    angle $\theta\in [0, 2\pi)$,   at $t\in \mathbb{R}$
is defined as the integral of $f$ along the line $(x,y)=
t(\cos\theta, \sin\theta)+s(-\sin\theta, \cos\theta):$
\begin{align}\label{gbhf} \int^{\infty}_{-\infty}
f(t\cos\theta-s\sin\theta, t\sin\theta+s\cos\theta)ds. \end{align}
Denote  $P=[\cos\theta, \sin\theta]$. For simplicity, following B. Han \cite{Han1} the  Radon transform in \eqref{gbhf} is denoted  as
$$Pf(t)=\int^{\infty}_{-\infty}
f(t\cos\theta-s\sin\theta, t\sin\theta+s\cos\theta)ds.$$
The Fourier transform of $Pf$  can be expressed as
 \begin{align}\label{Radonform} \widehat{Pf}(\xi)=\widehat{f}(P^{T}\xi)=\widehat{f}(\xi\cos\theta, \xi\sin\theta), \quad \xi\in \mathbb{R},\end{align}
where  for any function $g\in L^{p}(\mathbb{R}^{D})$ its Fourier transform at $\gamma\in \mathbb{R}^{D}$ is defined as $\widehat{g}(\gamma):=\int_{\mathbb{R}^{D}}g(x)e^{-\textbf{i}x\cdot \gamma}dx$.
It follows from \eqref{Radonform} that    $\widehat{Pf}$ is essentially
the  projection of $\widehat{f}$ onto the subspace (slice) $\{P^{T}\xi: \xi\in \mathbb{R}\}$.
Correspondingly, from now on $P$ is called
the \emph{projection  vector}.
On the other hand, $f$ can be reconstructed via the so called inverse Radon transform (c.f.\cite{Natterer}):
\begin{align}\label{pppp}
f(x,y)=\frac{1}{4\pi^{2}}\int^{2\pi}_{0}\int^{\infty}_{0}\widehat{f}(P^{T}\xi)e^{\textbf{i}\xi(x\cos\theta+y\sin\theta)}\xi d\xi d\theta,
\end{align}
with $P=[\cos\theta, \sin\theta].$

\subsection{Traditional  approaches conducted by Radon transforms at multiple angles}
Theoretically, \eqref{pppp} implies that  the reconstruction of $f$  requires   the  projections  $\widehat{f}(P^{T}\xi)$
  for all angle   $\theta$.
  In practice, however, what we can observe are the   samples of limited  projections.
Therefore, the essential problem of CT (also referred to as the  limited angle problem (LAP)) is to construct the approximation to $f$
by the  samples of  finitely many   projections. The most classical approximation  is  the filtered backprojection (FBP) (c.f. \cite{Kak,Shengli})
derived from \eqref{pppp}. In principle, FBP  requires $f$ to be bandlimited and needs the Fourier samples.
If  FBP is applied  to non-bandlimited functions, then the low-pass filtering (to cut  off the high-frequency component)  and linear interpolation are commonly  necessary. They    may  introduce
serious errors (c.f. \cite{YXu,Natterer1}).

Some recent  alternatives     to FBP have been introduced (e.g. \cite{Unser1,Kerkyacharian,Bai,Unser2,YXu}). Unlike FBP, they are conducted by the samples of
Radon transforms.  For example, based on the orthogonal polynomial system,  Xu
\cite{YXu} established the approach to CT.
Entezari, Nilchian and  Unser
\cite{Unser1} and  McCann   and  Unser \cite{Unser2} established a    spline-based  reconstruction.
Note  that the  samples required  for   the above approaches    are derived from  Radon transforms at  multiple angles.

Contrary to the traditional  approaches requiring multiple angles,
our purpose is to reconstruct  a function   by  the samples of its  Radon transform at a single angle (SA).
In what follows, we  introduce the requirement    for such  a   SA-based  approach in optical imaging.

\subsection{SA-based reconstruction is required  for real time imaging}
\label{fastimaging}
Optical imaging has been widely used in observing biological objects, such as blood cells (thin objects) and bones (thick objects). The thin objects are commonly imaged directly by refractive-index distributions, which  is  achieved by holographic tomography (HT) (\cite{Kim}). However, for imaging thick  objects,   CT  is  usually employed.

CT   commonly  requires  samples  (measurements) of the light fields penetrating through the object    from different angles (views).
To do so, the object  needs to be rotated by a rotation motor (\cite{Yablon}) or the illumination needs to be scanned by a beam steering device, which not only causes instability for the imaging system, but makes the system bulky (\cite{Antipa,Horisaki}). More importantly, limited by the time of recording fields, rotating objects or scanning illuminations becomes not suitable for real-time imaging, especially for observing fast dynamic events (\cite{Horisaki}).
Now a natural  imaging  problem is:  \begin{align} \label{tywm}  \hbox{ can     CT
be achieved by  the samples of SA-Radon transform?} \end{align}
Most recently,  R. Horisaki, K. Fujii, and J. Tanida  \cite{Horisaki} established a
SA method for HT  by  inserting a diffuser.
However, to our best of knowledge, \eqref{tywm} (the SA problem for CT)  is very challenging and remains  less explored.
The purpose of this paper is to investigate   it  from the mathematical perspective.
In what follows, we  organize the problem in the present paper.

\subsection{The mathematical problem of  SA-based reconstruction  in refinable SISs}
Recall that the  refinable  shift-invariant space (SIS) is  a   type of  function space that is  widely  used in signal processing
and time-frequency analysis
(e.g. \cite{Fienup289,Feris,Matla,SUNQIY1,yunzhang1}). On the other hand,  motivated
by the  connection between Bin Han's   construction of wavelet frames  (c.f.  \cite{Han1,Hanprojection,Hanprojectionold,Hanappletter} and to be introduced in subsection \ref{connection}) and  Radon transform \eqref{Radonform},
we will investigate the following  problem:
\begin{align} \label{SPP} \begin{array}{llll} \hbox{\textbf{Q}}: \hbox{Can one reconstruct a function in a refinable SIS by  its SA-Randon (transform)}\\
\quad \quad \hbox{samples}?\end{array}\end{align}
As the  problem  \eqref{tywm} in optics, the above  problem is still less explored in mathematics.
 It might be "stereotyped" by the  inverse Radon transform  \eqref{pppp} which requires the Fourier samples of  Radon transform at multiple angles.


\subsection{Our contributions}
Our   main results are organized in Theorems \ref{ee2} and   \ref{optimal}.
%
 In Theorem \ref{ee2} we  prove that,   any compactly supported
function   from  a general  refinable
SIS can be exactly reconstructed by its  Radon  samples at the appropriate angle.
 For a class of  refinable   box-splines, we   prove that their Radon transforms  at any angle are
 positive definite (PD). And  Theorem \ref{optimal} states
 that for  almost every     angle $\theta\in [0, 2\pi)$, any compactly supported function in such  box-splines  generated  SISs can be reconstructed by
 its  Radon (transform at   angle $\theta$)  samples.


\section{Preliminary}
\subsection{Connection between Han's    construction of wavelet frame and  Radon transform \eqref{Radonform}}\label{connection}
B. Han et.al \cite{Han1,Hanprojection,Hanprojectionold,Hanappletter} and \cite[Section 7.1.3]{BHBOOK}
provided the  projection  method for constructing  wavelet frames. Specifically,
the wavelet frames  $\{\psi^{1}, \ldots, \psi^{L}\}\subseteq L^{2}(\mathbb{R}^{2})$ are first constructed,
then the ones   $\{\tilde{P}\psi^{1}, \ldots, \tilde{P}\psi^{L}\}\subseteq L^{2}(\mathbb{R})$ defined by
\begin{align}\label{hansform}\widehat{\tilde{P}\psi^{l}}(\xi)=\widehat{\psi^{l}}(\tilde{P}^{T}\xi)\end{align}  can have desirable properties if choosing an  appropriate
 vector  $\tilde{P}\in \mathbb{Z}^{1\times 2}$. Note that both  \eqref{Radonform} and \eqref{hansform} take the identical form.
Moreover, since  $\widehat{\tilde{P}\psi^{l}}=\widehat{\psi^{l}}(\|\tilde{P}\|_{2}\frac{\tilde{P}^{T}}{\|\tilde{P}\|_{2}}\xi)$ then $\tilde{P}\psi^{l}$
is essentially the composition of  Radon transform (projection vector: $\frac{\tilde{P}^{T}}{\|\tilde{P}\|_{2}}$)  and the  dilation operation (dilation factor: $\|\tilde{P}\|^{-1}_{2}$).

\subsection{On the  support of  $\tilde{P}f$ defined via \eqref{hansform}}
For a function $f\in  L^{1}(\mathbb{R}^{2})$ and a   vector  $\tilde{P}\in \mathbb{R}^{1\times 2}$,
 motivated by  \cite{Han1,BHBOOK} we next address the relationship  between     $\tilde{P}f$ and $f$ in the  spatial  domain. Denote the singular value decomposition (SVD) of $\tilde{P}$
by $\tilde{P}=\Sigma V^{T}$, where  $V$ is a  $2\times 2$ real-valued unitary
matrix, and $\Sigma=[
\sigma, 0]$ with $\sigma=\|P\|_{2}$.
Now it follows from     \cite{Han1,BHBOOK} that
$
\tilde{P}f=\Sigma(V^{T}f),
$
where $V^{T}f(x)=f((V^{T})^{-1}x)$ with  $x=(x_{1},   x_{2})^{T}\in \mathbb{R}^{2}$,  and for any $g$ on $\mathbb{R}^{2}$ the function  $\Sigma g$ is defined by
\begin{align}\label{cvb}
\Sigma g (x_{1})=\sigma^{-1}\int_{\mathbb{R}} g(\sigma^{-1}x_{1},
x_{2})dx_{2}.
\end{align}
It is easy to check that if  $f\in L^1(\mathbb{R}^2)$ then $\tilde{P}f\in L^1(\mathbb{R})$.
Moreover,  the following lemma states that     if $f\in
L^2(\mathbb{R}^2)$ is compactly supported,   then $\tilde{P}f$ still lies in $ L^2(\mathbb{R})$ and is also compactly supported.

\begin{lem}\label{pouy}
Suppose that $f\in L^{2}(\mathbb{R}^{2})$ is compactly supported on $[a_{1}, b_{1}]\times
[a_{2}, b_{2}]$. For $0\neq\tilde{P}\in \mathbb{R}^{1\times 2}$, define
 $\tilde{P}f$ via \eqref{hansform}.
Then
\begin{align}\label{qujian} \hbox{supp}(\tilde{P}f)\subseteq [-\sqrt{2}\|\tilde{P}\|_{2}\max\{|b_{i}|, |a_{i}|: i=1, 2\}, \sqrt{2}\|\tilde{P}\|_{2}\max\{|b_{i}|, |a_{i}|: i=1, 2\}],\end{align}
and $\tilde{P}f\in L^{2}(\mathbb{R})$.
\end{lem}
\begin{proof}
Denote $V^{T}f$ by $g.$ Then for any $x_{1}\in \mathbb{R},$
\begin{align}\label{en123456}
\begin{array}{llll}
\tilde{P}f(x_{1})&=\displaystyle\Sigma g(x_{1})\\
&\displaystyle=\sigma^{-1}\int_{\mathbb{R}} g(\sigma^{-1}x_{1},
x_{2})dx_{2}\\
&\displaystyle=\sigma^{-1}\int_{\mathbb{R}} [V^{T}f](\sigma^{-1}x_{1},
x_{2})dx_{2}\\
&\displaystyle=\sigma^{-1}\int_{\mathbb{R}} f(V(\sigma^{-1}x_{1},
x_{2})^{T})dx_{2},
\end{array}
\end{align}
where $\sigma=\|\tilde{P}\|_{2}$, the second and last  equalities are  derived from \eqref{cvb} and $V$ being a real-valued  unitary
matrix.
It follows from   $\hbox{supp}(f)\subseteq [a_{1}, b_{1}]\times [a_{2}, b_{2}]$ that  for any $x\in \hbox{supp}(f)$, we have
$\|x\|_{2}\leq\sqrt{2}\max\{|b_{i}|, |a_{i}|: i=1, 2\}$. Now  by  \eqref{en123456} and
$V$ being a unitary matrix it is easy to check that  \eqref{qujian} holds.
Define $G(x_{1}, x_{2}):=f(V(x_{1}, x_{2})^{T})$. For any $(x_{1}, x_{2})^{T}\in \hbox{supp}(G)$,
through the similar analysis   as above we have
\begin{align}\label{uytt} |x_{2}|\leq
\sqrt{2}\max\{|b_{i}|, |a_{i}|: i=1, 2\}.\end{align}
Moreover,
\begin{align}\label{en123}
\begin{array}{llll}
\|\tilde{P}f\|_{L^{2}}^{2}&=\|\Sigma(V^{T}f)\|_{L^{2}}^{2}\\
&\displaystyle=\frac{1}{\sigma^{2}}\int_{\mathbb{R}}\big|\int_{\mathbb{R}} [V^{T}f](\sigma^{-1}x_{1},
x_{2})dx_{2}\big|^{2} dx_{1}\\
&\displaystyle\leq\frac{\sqrt{2}\max\{|b_{i}|, |a_{i}|: i=1, 2\}}{\sigma^{2}}\int_{\mathbb{R}}\int_{\mathbb{R}} |[V^{T}f](\sigma^{-1}x_{1},
x_{2})|^{2}dx_{2}dx_{1}\\
&\displaystyle\leq\frac{\sqrt{2}\max\{|b_{i}|, |a_{i}|: i=1, 2\}}{\sigma}\|f\|_{L^{2}}^{2}<\infty,
\end{array}
\end{align}
where the first inequality is derived from \eqref{uytt} and  the Cauchy-Swart inequality.
The proof is concluded.
\end{proof}

\subsection{(Quasi) shift-invariant space}
\label{Rieszbound}
The space of  square summable sequences on $\mathbb{Z}^{D}$
is defined  as $\ell^{2}(\mathbb{Z}^{D}):=\big\{ \{c_{k}\}_{k\in \mathbb{Z}^{D}}: \|\{c_{k}\}_{k\in \mathbb{Z}^{D}}\|_{\ell^{2}}=[\sum_{k\in \mathbb{Z}^{D}}|c_{k}|^{2}]^{1/2}<\infty\big\}$.
For a generator    $\phi\in  L^{2}(\mathbb{R}^{D})$  with $D\geq1$, the associated  shift-invariant space (SIS) $V(\phi, \mathbb{Z}^{D})$ is defined as
\begin{align}\label{SISDE} V(\phi, \mathbb{Z}^{D})=\Big\{\sum_{k\in \mathbb{Z}^{D}}c_{k}\phi(\cdot-k): \{c_{k}\}_{k\in \mathbb{Z}^{D}}\in \ell^{2}(\mathbb{Z}^{D})\Big\}.\end{align}
For the stable recovery of functions in $V(\phi, \mathbb{Z}^{D})$, it is commonly required that
$\big\{\phi(\cdot-k): k\in \mathbb{Z}^{D}\big\}$ is a Riesz basis for $V(\phi, \mathbb{Z}^{D})$ (c.f. \cite{Fienup289,Aldroubi3,Jun}), namely, there exist $0<C_{1}\leq C_{2}<\infty$
such that
\begin{align}\label{chuan} C_{1}\|\{c_{k}\}_{k\in \mathbb{Z}^{D}}\|^{2}_{\ell^{2}}\leq\big\|\sum_{k\in \mathbb{Z}^{D}}c_{k}\phi(\cdot-k)\big\|^{2}_{L^{2}}\leq C_{2}\|\{c_{k}\}_{k\in \mathbb{Z}^{D}}\|^{2}_{\ell^{2}}.\end{align}
Here $C_{1}$ and $C_{2}$ are referred to as the Riesz  lower and upper bounds  of  $\{\phi(\cdot-k): k\in \mathbb{Z}^{D}\}$.
For a generator $\varphi\in  L^{2}(\mathbb{R}^{D})$ and the shift set  $\mathcal{X}=\{x_{j}\}_{j\in \mathbb{Z}^{D}}\subseteq \mathbb{R}^{D}$, the associated quasi shift-invariant space  (QSIS) is defined as
\begin{align}\label{QSIS}V(\varphi, \mathcal{X}):=\Big\{\sum_{k\in \mathbb{Z}^{D}}c_{k}\varphi(\cdot-x_{j}): \{c_{k}\}_{k\in \mathbb{Z}^{D}}\in \ell^{2}(\mathbb{Z}^{D})\Big\}.\end{align}
If $\mathcal{X}=\mathbb{Z}$ then $V(\varphi, \mathcal{X})$ degenerates to be  an SIS.
As implied  in \cite{QSIS4}, the recovery  theory for  the QSIS ($\mathcal{X}\neq\mathbb{Z}$ )
is not the trivial generalization of that for  the SIS.  For example, the generator for the QSIS
needs to be a  positive definite function  (the definition  is  postponed to section \ref{boxspline}) such that the recovery can be achieved (c.f. \cite[section 3.1(A1)]{QSIS4}).
Readers can refer  to  \cite{QSIS1,QSIS3,QSIS2,QSIS4} for the recent   recovery  results  in
QSISs (generated from  some interesting generators such as the \hbox{sinc}   and Gaussian functions).

\subsection{On the Sobolev smoothness of a function}
For any $\varsigma\in\mathbb{R}$,  the Sobolev space $H^{\varsigma}(\mathbb{R}^{D})$
is defined as
\begin{align}\begin{array}{lllll}\label{defi_s}\displaystyle  H^{\varsigma}(\mathbb{R}^{D}):=\Big\{f: \int_{\mathbb{R}^{D}}|\widehat{f}(\xi)|^{2} (1+\|\xi\|_{2}^{2})^{\varsigma}d\xi<\infty\Big\}.\end{array}\end{align}
 The deduced norm is defined   by
\begin{align}\notag \begin{array}{lllll} \displaystyle \|f\|_{H^{\varsigma}(\mathbb{R}^{D})}:=\frac{1}{(2\pi)^{D/2}}\Big(\int_{\mathbb{R}^{D}}|\widehat{f}(\xi)|^{2}(1+\|\xi\|_{2}^{2})^{\varsigma}d\xi\Big)^{1/2}, \quad \forall f\in H^{\varsigma}(\mathbb{R}^{D}).\end{array}\end{align}
And the Sobolev smoothness of $f$ is defined as $\nu_{2}(f):=\sup\{\varsigma: f\in H^{\varsigma}(\mathbb{R}^{D})\}$.
For $\varsigma>D/2$, by  \cite[Chapter 9.1]{Matla} or \cite[section 1]{LHYH}
the functions in  $H^{\varsigma}(\mathbb{R}^{D})$   are continuous.

\section{SA-Radon samples based  reconstruction    for  compactly supported functions in  general refinable SISs}\label{SPPFORSIS}
Some definitions and denotations are necessary for our discussion. For  a SIS $V(\phi, \mathbb{Z}^{D})$ and a compact set $E\subseteq\mathbb{R}^D $,
define $$V_{E}(\phi, \mathbb{Z}^{D}):=\{f\in V(\phi, \mathbb{Z}^{D}):   \hbox{supp}(f)\subseteq E\}.$$
 Denote by $\lfloor x\rfloor$    ($\lceil x\rceil$)
the largest (smallest) number that is not larger  (not smaller) than $x$. The cardinality of a set $G$  is
denoted by $\#G$. A   function $\phi\in L^{2}(\mathbb{R}^{D})$
is refinable if $\widehat{\phi}(2\xi)=\widehat{a}(\xi)\widehat{\phi}(\xi)$ where
$\widehat{a}(\xi)$ is a  $2\pi\mathbb{Z}^{D}$-periodic trigonometric polynomial (c.f. \cite{Chui,BHBOOK,Youfa}).
In what follows we establish the first main theorem in this paper.

\subsection{The first main theorem}
The following is the first main result in this paper.
\begin{theo}\label{ee2}
Suppose that  $\phi\in L^{2}(\mathbb{R}^{2})$ is refinable and continuous   such that $\hbox{supp}(\phi)\subseteq[N_{1}, M_{1}]\times
[N_{2}, M_{2}]$.
Moreover,  suppose that  the projection vector $P\in \mathbb{R}^{1\times 2}$  and there exists    $\gamma\in \mathbb{R}$  such that $\gamma P\in \mathbb{Z}^{1\times 2}$ and  $\{\gamma Px: x\in \mathbb{Z}^{2}\}=\mathbb{Z}$.
Then for  almost all  discrete set $X\subseteq(0, 1)$ such that $2\leq\#X< \infty$,   any $f\in V_{E}(\phi, \mathbb{Z}^{2})$ with  $E=[a_{1}, b_{1}]\times [a_{2}, b_{2}]$ can be determined
 by the     samples  of its SA-Radon transform   $Pf$ at
 \begin{align}\label{9876tg22}\begin{array}{lllllllll}
 \big(X/\gamma+\mathbb{Z}/\gamma\big)\\
 \bigcap \Big[-\sqrt{2}\max\{|N_{i}|, |M_{i}|: i=1, 2\}+K_{\min}, \sqrt{2}\max\{|N_{i}|, |M_{i}|: i=1, 2\}+K_{\max}\Big],
 \end{array}
 \end{align}
   provided that
the map $P_{\mathring{E}}: \mathring{E}\longrightarrow \mathbb{Z}$ defined by $P_{\mathring{E}}(x)=Px$ is an  injection,
  where \begin{align}\label{eexyz} \mathring{E}=\Big\{\Big[\lceil a_{1}-M_{1}\rceil, \lfloor b_{1}-N_{1}\rfloor\Big]\times
\Big[\lceil a_{2}-M_{2}\rceil, \lfloor b_{2}-N_{2}\rfloor\Big]\Big\}\cap \mathbb{Z}^{2},\end{align}
$K_{\min}=\min\{\gamma Pk: k\in \mathring{E}\}$ and
 $K_{\max}=\max\{\gamma Pk: k\in \mathring{E}\}$.
\end{theo}
\begin{proof}
The proof is given in subsection \ref{pppfg}.
\end{proof}

\begin{rem}
In Theorem \ref{ee2},  we assume that  the functions to be reconstructed  are compactly
 supported and their supports are contained in a known rectangle.  Such an   assumption     is reasonable  for CT (c.f. \cite{YXu}).
\end{rem}

\subsection{Design of feasible  projection vector}\label{conditionnn}
In  Theorem  \ref{ee2}  the projection  vector  $P$ is required to satisfy
the following two conditions:

$(A_{1}):$ the map $P_{\mathring{E}}: \mathring{E}\longrightarrow \mathbb{Z}$ defined by $P_{\mathring{E}}(x)=Px$ is  injective.

$(A_{2}):$ there exists $0\neq\gamma\in \mathbb{R}$ such that  $\{\gamma Px: x\in \mathbb{Z}^{2}\}=\mathbb{Z}.$

In what follows, we establish a choice of $P$ such that $(A_{1})$ and $(A_{2})$ hold.

\begin{prop}\label{choiceofprojection}
Define $\mathring{E}^{-}:=\{k-\hat{k}: k, \hat{k}\in \mathring{E}\}$.
For $j=1,2$, denote $\mathring{E}^{-}_{j, \max}:=\{|\langle k, e_{j}\rangle|: k\in \mathring{E}^{-}\}$,
where $e_{j}\in \mathbb{R}^{2}$ is the $j$th unit coordinate vector with the only nonzero entry $1$ at
the $j$th entry.
Choose  $\tilde{P}=[\tilde{P}_{1},  \tilde{P}_{2}]$   recursively  by $\tilde{P}_{1}=1$, $\tilde{P}_{2}=\mathring{E}^{-}_{1, \max}+\tilde{P}_{1}$.  Then $P=\frac{\tilde{P}}{\|\tilde{P}\|_{2}}$ satisfies $(A_{1})$ and $(A_{2}).$
\end{prop}
\begin{proof}
It is easy to prove that  $(A_{1})$ holds if and only if
$\langle \tilde{P},  k\rangle\neq0$ for any fixed  $\textbf{0}\neq k=(k_{1},   k_{2})\in \mathring{E}^{-}.$
Denote  $J_{k}=\max\{j: k_{j}\neq0\}$.
If $J_{k}=1$, then clearly $|\langle \tilde{P},  k\rangle|\geq1$.
If $J_{k}=2$, then   $$|\langle \tilde{P},  k\rangle|\geq |\tilde{P}_{2}k_{2}|-|\tilde{P}_{1}k_{1}|=
|\tilde{P}_{2}k_{2}|-|k_{1}|\geq |\tilde{P}_{2}|-|k_{1}|\geq |\tilde{P}_{2}|-\mathring{E}^{-}_{1, \max}=1.$$
Then $(A_{1})$ holds.
Since $\langle \tilde{P},  e_{1}\rangle=\tilde{P}_{1}=1$ then  $(A_{2})$ holds with $\gamma=\|\tilde{P}\|_{2}$.
\end{proof}

\subsection{Proof of Theorem \ref{ee2}}\label{pppfg}
\subsubsection{Two lemmas}

In the following lemma, we establish a necessary and sufficient condition on the generator Radon transform  $P\phi$,
such that any compactly supported $f\in V(\phi, \mathbb{Z}^{2})$ can be determined by its  SA-Radon transform  $Pf.$

\begin{lem}\label{main1}
Suppose that  $\phi\in  L^{2}(\mathbb{R}^{2})$ is refinable  such that    $\hbox{supp}(\phi)\subseteq[N_{1}, M_{1}]\times
[N_{2}, M_{2}]$.
Then  any $f\in V_{E}(\phi, \mathbb{Z}^{2})$ can be determined by its  Radon transform   $Pf$ if and only if
$\{P\phi(\cdot-Pk): k\in \mathring{E}\}$
is a basis for  $\hbox{span}\{P\phi(\cdot-Pk): k\in \mathring{E}\}$, where $E=[a_{1}, b_{1}]\times [a_{2}, b_{2}]$ and
\begin{align}\label{ee1} \mathring{E}=\Big\{\Big[\lceil a_{1}-M_{1}\rceil, \lfloor b_{1}-N_{1}\rfloor\Big]\times
\Big[\lceil a_{2}-M_{2}\rceil, \lfloor b_{2}-N_{2}\rfloor\Big]\Big\}\cap \mathbb{Z}^{2}.\end{align}
%
%
\end{lem}
\begin{proof}
Necessity: For   $f\in V_{E}(\phi, \mathbb{Z}^{2})$, it follows from    $\hbox{supp}(f)\subseteq E$ that
\begin{align}\label{XK} f=\sum^{\lfloor b_{1}-N_{1}\rfloor}_{k_{1}=\lceil a_{1}-M_{1}\rceil} \sum^{\lfloor b_{2}-N_{2}\rfloor}_{k_{2}=\lceil a_{2}-M_{2}\rceil}c_{k_{1},  k_{2}}\phi(\cdot-k_{1},   \cdot-k_{2}).\end{align}
Then for any $\xi\in \mathbb{R}$, we have
\begin{align}\begin{array}{llll} \label{representation}\widehat{Pf}(\xi)&=\widehat{f}(P^{T}\xi)&=\displaystyle\sum^{\lfloor b_{1}-N_{1}\rfloor}_{k_{1}=\lceil a_{1}-M_{1}\rceil} \sum^{\lfloor b_{2}-N_{2}\rfloor}_{k_{2}=\lceil a_{2}-M_{2}\rceil}c_{k_{1},  k_{2}}e^{-\textbf{i}k\cdot P^{T}\xi}\widehat{\phi}(P^{T}\xi)\\
&&=\displaystyle\sum^{\lfloor b_{1}-N_{1}\rfloor}_{k_{1}=\lceil a_{1}-M_{1}\rceil} \sum^{\lfloor b_{2}-N_{2}\rfloor}_{k_{2}=\lceil a_{2}-M_{2}\rceil}c_{k_{1},  k_{2}}e^{-\textbf{i}Pk\cdot\xi}\widehat{P\phi}(\xi),\end{array}\end{align}
where $k=(k_{1},  k_{2})^{T}$. That is,
\begin{align}\label{YK} Pf=\sum^{\lfloor b_{1}-N_{1}\rfloor}_{k_{1}=\lceil a_{1}-M_{1}\rceil} \sum^{\lfloor b_{2}-N_{2}\rfloor}_{k_{2}=\lceil a_{2}-M_{2}\rceil}c_{k_{1},  k_{2}}P\phi(\cdot-Pk)\in \hbox{span}\{P\phi(\cdot-Pk): k\in \mathring{E}\}.\end{align}
If $\{P\phi(\cdot-Pk): k\in \mathring{E}\}$
is a basis for $\hbox{span}\{P\phi(\cdot-Pk): k\in \mathring{E}\}$, then the map  $P_{\mathring{E}}: \mathring{E}\longrightarrow Px, x\in \mathring{E}$
is an  injection. For convenient narration we  denote $\{Px: x\in \mathring{E}\}$ by $\Lambda$.
Then there exists uniquely  $\{\widehat{c}_{j}\}_{j\in \Lambda}$ such that   $Pf=\sum_{j\in \Lambda}\widehat{c}_{j}P\phi(\cdot-j)$.
Comparing \eqref{XK} and \eqref{YK}, we have $f=\sum_{j\in \Lambda}\widehat{c}_{j}\phi(\cdot-P^{-1}_{\mathring{E}}(j))$.

Sufficiency:  Without loss of generality, suppose  that the dimension  $\hbox{dim}(\hbox{span}\{P\phi(\cdot-Pk): k\in \mathring{E}\})\geq1.$
If $\{P\phi(\cdot-Pk): k\in \mathring{E}\}$
is not a basis for $\hbox{span}\{P\phi(\cdot-Pk): k\in \mathring{E}\}$, then
 there exists a nonzero sequence  $\{c_{k}\}_{k\in \mathring{E}}$
 such that $0\equiv\sum_{k\in \mathring{E}}c_{k}P\phi(\cdot-Pk)$.
 Therefore, $0\not\equiv\widetilde{f}:=\sum_{k\in \mathring{E}}c_{k}\phi(\cdot-k)$ is not distinguishable from $ g\equiv0\in V_{E}(\phi, \mathbb{Z}^{D})$
since they have the same projection. This leads to a contradiction.
%
%
%
\end{proof}

\begin{lem}\label{uyt701}
Suppose that  $\phi\in  L^{2}(\mathbb{R}^{2})$ is refinable  and $\tilde{P}\in \mathbb{Z}^{1\times 2}$ such that $\{\tilde{P}x: x\in \mathbb{Z}^{2}\}=\mathbb{Z}$.
Then $\tilde{P}\phi$ defined by $\widehat{\tilde{P}\phi}(\xi)=\widehat{\phi}(\tilde{P}^{T}\xi)$ is compactly supported  and  refinable    in $L^{2}(\mathbb{R}).$
\end{lem}
\begin{proof}
By Lemma \ref{pouy}, $\tilde{P}\phi$ is compactly supported. By \cite[Theorem 7.1.9]{BHBOOK},
if $\tilde{P}^{T}(\mathbb{Z}\setminus[2\mathbb{Z}])\subseteq 2\mathbb{Z}^{2}$ then
$\tilde{P}\phi$ is still  refinable such that $\{\tilde{P}\phi(\cdot-k)\}_{k\in \mathbb{Z}}$ is linearly independent. Since $\{\tilde{P}x: x\in \mathbb{Z}^{2}\}=\mathbb{Z},$
then it follows from \cite[Proposition 4.2]{Hanprojection} that $\tilde{P}^{T}(\mathbb{Z}\setminus[2\mathbb{Z}])\subseteq 2\mathbb{Z}^{2}$
is satisfied. Now the proof can be easily  concluded.
\end{proof}


\subsubsection{Proof of Theorem \ref{ee2}}
Denote $\tilde{P}=\gamma P$. It is straightforward to check that $Pf(x)=\gamma\tilde{P}f(\gamma x)$.
As in \eqref{YK}, we have
\begin{align}\label{YK1} \tilde{P}f=\sum^{\lfloor b_{1}-N_{1}\rfloor}_{k_{1}=\lceil a_{1}-M_{1}\rceil} \sum^{\lfloor b_{2}-N_{2}\rfloor}_{k_{2}=\lceil a_{2}-M_{2}\rceil}c_{k_{1},  k_{2}}\tilde{P}\phi(\cdot-\tilde{P}k)\in \hbox{span}\{\tilde{P}\phi(\cdot-\tilde{P}k): k\in \mathring{E}\}.\end{align}
It follows from Lemma \ref{uyt701} that $\tilde{P}\phi$ is compactly supported  and  refinable    in $L^{2}(\mathbb{R}).$
Moreover, by \cite[Lemma 2.4]{Hanprojectionold} we have $\nu_{2}(\tilde{P}\phi)\geq \nu_{2}(\phi)\geq 1$,  which together with   \cite[Chapter 9.1]{Matla} (or \cite[section 1]{LHYH}) leads to that   $P\phi$
 is continuous.  On the other hand, since $\{\tilde{P}x: x\in \mathbb{Z}^{2}\}=\mathbb{Z}$ then
$\tilde{P}f\in V(\tilde{P}\phi, \mathbb{Z})$.
 Now  it follows from  \cite[section 4]{SUNQIY1} that $\tilde{P}f$ (or the sequence $\{c_{k_{1},  k_{2}}\}$ in \eqref{YK1}) can determined by    its  samples  at $X+\mathbb{Z}$.  Since the map  $P_{\mathring{E}}: \mathring{E}\longrightarrow Px, x\in \mathring{E}$
is injective, then so is the map $\tilde{P}_{\mathring{E}}: \mathring{E}\longrightarrow \tilde{P}x, x\in \mathring{E}$.
Now by \eqref{YK1} and   Lemma \ref{uyt701}, $Pf$ can be rewritten as
 \begin{align}\label{08765} Pf=\sum_{j\in \Lambda}\widehat{c}_{j}P\phi(\cdot-j),\end{align}
where $\Lambda=\{\tilde{P}x: x\in \mathring{E}\}$ and $\{\widehat{c}_{j}: j\in \Lambda\}$
is a rearrangement of $\{c_{k_{1},  k_{2}}\}$.
 Denote $K_{\min}=\min\{\tilde{P}k: k\in \mathring{E}\}$ and
 $K_{\max}=\max\{\tilde{P}k: k\in \mathring{E}\}$.
 Then  \begin{align}\label{iuytg} \Lambda\subseteq\{
 K_{\min}, K_{\min}+1, \ldots, K_{\max}\}.\end{align}
 By Lemma \ref{pouy}, we have
 \begin{align}\label{qujian0098} \hbox{supp}(\tilde{P}\phi)\subseteq \big[-\sqrt{2}|\gamma|\max\{|N_{i}|, |M_{i}|: i=1, 2\}, \sqrt{2}|\gamma|\max\{|N_{i}|, |M_{i}|: i=1, 2\}\big].\end{align}
 Now it follows from \eqref{08765},  \eqref{iuytg} and \eqref{qujian0098} that  $\tilde{P}f$ can be determined by its samples at
 \begin{align}\label{9876tg}\begin{array}{lllllllll}
 \big(X+\mathbb{Z}\big)\\
 \bigcap \Big[-\sqrt{2}|\gamma|\max\{|N_{i}|, |M_{i}|: i=1, 2\}+K_{\min}, \sqrt{2}|\gamma|\max\{|N_{i}|, |M_{i}|: i=1, 2\}+K_{\max}\Big].
 \end{array}
 \end{align}
 Equivalently, $Pf$ can be determined by its samples at
 \begin{align}\label{9876tg2}\begin{array}{lllllllll}
 \big(X/\gamma+\mathbb{Z}/\gamma\big)\\
 \bigcap \big[-\sqrt{2}\max\{|N_{i}|, |M_{i}|: i=1, 2\}+K_{\min}, \sqrt{2}\max\{|N_{i}|, |M_{i}|: i=1, 2\}+K_{\max}\big].
 \end{array}
 \end{align}
 The proof is concluded by Lemma \ref{main1}.

\section{Reconstruction   of    functions in    SISs generated from a class of refinable  splines by SA-Radon samples}
\label{boxspline}
In this section, we are interested in the SA problem \eqref{SPP}  for   compactly supported  functions  in the SISs   generated by  a class of refinable  box-splines. The main result will be  organized  in Theorem \ref{optimal}. Due to
the positive definite (PD) property of
the Radon transform at any angle of such splines, the theorem  enjoys great angle  flexibility.
In what follows, we   introduce     the  PD functions and the determination  of functions in their generated QSISs.

\subsection{PD functions and the determination  of functions in the  generated QSISs}
We start with the  definition of a PD  function
which is extensively   applied to interpolation of scattered data (c.f. \cite{Karimi,Wendlan}).
\begin{defi}\label{definitionofPD}
A function $\varphi: \mathbb{R}^{D}\longrightarrow \mathbb{C}$
is  PD if, for all $N\in \mathbb{N}$, all sets $X=\{x_{1}, x_{2}, \ldots, x_{N}\}\subseteq \mathbb{R}^{D}$ of pairwise distinct elements, and all $0\neq \alpha=(\alpha_{1}, \ldots, \alpha_{N})\in \mathbb{C}^{N}$, the quadratic form
\begin{align}\label{zhengding}
\sum^{N}_{j=1}\sum^{N}_{k=1}\alpha_{j}\overline{\alpha}_{k}\varphi(x_{j}-x_{k})>0.
\end{align}
\end{defi}

The following proposition  concerns on the determination  of functions by  the PD property. It can be easily
proved by  \eqref{zhengding}.

\begin{prop}\label{6758}
If $\varphi: \mathbb{R}^{D}\longrightarrow \mathbb{C}$ is PD, then $\varphi(0)>0,$ and for any set  $X=\{x_{1}, x_{2}, \ldots, x_{N}\}\subseteq \mathbb{R}^{D}$,   the system   $\{\varphi(\cdot-x_{k})\}^{N}_{k=1}$ is linearly independent. Moreover, any function $f=\sum^{N}_{k=1}\alpha_{k}\varphi(\cdot-x_{k})$ can be determined by its samples  at $X$.
\end{prop}

Next we  introduce  why we require the  PD property  for the  SA problem \eqref{SPP}.
Recall that in the  proof of   Theorem \ref{ee2},
the reconstruction  of $f\in V(\phi, \mathbb{Z})$ depends on  that of its   Radon transform  $Pf$.
For the projection vector $P$, as in section \ref{conditionnn} (A1) it is required that
there exists $\gamma\in \mathbb{R}$ such that  $\tilde{P}:=\gamma P\in \mathbb{Z}^{1\times 2}$.
The condition $\tilde{P}\in \mathbb{Z}^{1\times 2}$ guarantees that $\tilde{P}\phi$ is still   refinable   and
$\tilde{P}f$ lies in the SIS $ V(\tilde{P}\phi, \mathbb{Z})$. By the relationship $Pf(x)=\gamma\tilde{P}f(\gamma x)$, we can
reconstruct $Pf$ by the   sampling method in the    SIS $V(\tilde{P}\phi, \mathbb{Z})$.
For  a general  refinable function
$\phi\in L^{2}(\mathbb{R}^{2})$, however,  if     $P$  does not satisfy    (A1)  then
Theorem \ref{ee2} probably  fails. In this section we will investigate the SA problem \eqref{SPP} for the SISs generated
by a class of   refinable box-splines. It will be clear in Theorem \ref{optimal} that   (A1)
is not necessary for  such  refinable functions  since their Radon transforms at any angle  are  PD, which admits the reconstruction in the QSISs by Proposition \ref{6758}. Stated another way, the reconstruction will have
great flexibility for choosing the angle of Radon transform.

\subsection{The second main theorem}
We first introduce   a class of refinable box-splines which admit the PD property of their Radon transforms at any angle.
The $m$th cardinal B-spline $B_{m}$ is defined
by
$B_{m}:=\overbrace{\chi_{(0,1]}\star\ldots\star\chi_{(0,1]}}^{m\ \hbox{copies}}$,
where $m\in \mathbb{N}$,  $\chi_{(0,1]}$ is the characteristic function of $(0,1],$
and $\star$ is the convolution.
Through the  simple calculation we have $\hbox{supp}(B_{m})=[0, m]$, and   $\widehat{B_{m}}(\xi)=e^{-\textbf{i} m\xi/2}[\frac{\sin\xi/2}{\xi/2}]^{m} $
by which it is easy to check that
$B_{m}$ is refinable (c.f. \cite{Chui}).
Define $\varphi_{2n}=B_{2n}(\cdot+n)$, $n\in \mathbb{N}$. Then $\widehat{\varphi_{2n}}(\xi)=[\frac{\sin\xi/2}{\xi/2}]^{2n}$.
Moreover, through the  tensor product  we  define $\phi_{B}: \mathbb{R}^{2}\longrightarrow \mathbb{R}$ by
\begin{align}\label{eeee17845} \phi_{B}(x_{1},  x_{2})=\prod^{2}_{k=1}\varphi_{2n_{k}}(x_{k}), \ n_{k}\in \mathbb{N}.
\end{align}
Clearly, \begin{align} \label{oppp} \widehat{\phi_{B}}(\xi_{1}, \xi_{2})=
\prod^{2}_{k=1}\Big[\frac{\sin\xi_{k}/2}{\xi_{k}/2}\Big]^{2n_{k}}\end{align} and $\phi_{B}$ is still refinable.
Moreover, $\hbox{supp}(\phi_{B})=[-n_{1}, n_{1}]\times [-n_{2}, n_{2}]$.

\begin{prop}\label{ertyfc}
Let $\phi_{B}$ be as in \eqref{eeee17845}. Then   for any projection vector $\textbf{0}\neq P\in \mathbb{R}^{1\times 2}$, the Radon transform  $P\phi_{B}$
is PD and even.
\end{prop}
\begin{proof}
The proof is given in section \ref{kjcd}.
\end{proof}

In what follows, we establish the main result in this section.

\begin{theo}\label{optimal}
Suppose that  $\phi_{B}\in L^{2}(\mathbb{R}^{2})$ is defined  in \eqref{eeee17845},  and  denote  $\hbox{supp}(\phi_{B})= [N_{1}, M_{1}]\times
[N_{2}, M_{2}]$.
Let  $E=[a_{1}, b_{1}]\times [a_{2}, b_{2}]$ and
\begin{align}\label{ee12} \mathring{E}=\Big\{\Big[\lceil a_{1}-M_{1}\rceil, \lfloor b_{1}-N_{1}\rfloor\Big]\times
\Big[\lceil a_{2}-M_{2}\rceil, \lfloor b_{2}-N_{2}\rfloor\Big]\Big\}\cap \mathbb{Z}^{2}.\end{align}
 Then for almost all    projection vector   $P\in \mathcal{B}:=\{\textbf{x}\in \mathbb{R}^{1\times 2}: \|\textbf{x}\|_{2}=1\}$,
any $f\in V(\phi_{B}, \mathbb{Z}^{2})$ such that $\hbox{supp}(f)\subseteq E$
  can be determined by its  SA-Radon   samples  $\{Pf(Pk): k\in \mathring{E}\}$.
%
%
\end{theo}

\begin{proof}
The proof is given in section \ref{optimalll}.
\end{proof}

\begin{rem}
Unlike Theorem \ref{ee2}, the projection vector $P$ in Theorem \ref{optimal} does not necessarily
satisfy the conditions $A_{1}$ and $A_{2}$ in section \ref{conditionnn}.
\end{rem}

\subsection{The proof of Proposition \ref{ertyfc}}\label{kjcd}
The following lemma is derived from \cite[section 6.2]{Wendlan}.
\begin{lem}\label{uyt7}
Suppose that $\varphi\in C(\mathbb{R})\cap L^{1}(\mathbb{R})$ such that $\widehat{\varphi}\in L^{1}(\mathbb{R})$.
Then for any set $\{x_{1}, \ldots, x_{n}\}$ of   pairwise distinct points on $\mathbb{R}$, it holds that
\begin{align}\label{eeee789} \sum^{n}_{k=1}\sum^{n}_{l=1}c_{k}\overline{c}_{l}\varphi(x_{k}-x_{l})=\int_{\mathbb{R}}\widehat{\varphi}(\xi)|\sum^{n}_{k=1}
c_{k}e^{-\textbf{i}x_{k}\xi}|^{2}d\xi,
\end{align}
where $\{c_{1}, \ldots, c_{n}\}\subseteq \mathbb{C}.$
\end{lem}
Based on Lemma \ref{uyt7} we establish the following lemma.
\begin{lem}\label{uyt8}
Suppose that $\varphi\in C(\mathbb{R})$ is even and  real-valued   such that $\widehat{\varphi}\in L^{1}(\mathbb{R})$.
If   $\widehat{\varphi}(\xi)\geq0$ for any $\xi\in \mathbb{R}$ and there exists $\zeta>0$
such that $\widehat{\varphi}(\xi)>0$  on $[-\zeta, \zeta]$,
then   the matrix
\begin{align}\label{qixiang}
A_{x_{1}, \ldots, x_{n}}=\left[\begin{array}{cccccccccc}
\varphi(0)&\varphi(x_{1}-x_{2})&\cdots&\varphi(x_{1}-x_{n})\\
\varphi(x_{2}-x_{1})&\varphi(0)&\cdots&\varphi(x_{2}-x_{n})\\
\vdots&\vdots&\ddots&\vdots\\
\varphi(x_{n}-x_{1})&\varphi(x_{n}-x_{2})&\cdots&\varphi(0)
\end{array}\right]
\end{align}
is positive definite for any set  $X=\{x_{1}, x_{2}, \ldots, x_{n}\}\subseteq \mathbb{R}$ of pairwise distinct points.
More precisely, there exists a positive constant $\Lambda$ (being dependent on $\varphi, X$
and $\zeta$) such that for any vector $\textbf{0}\neq(c_{1}, \ldots, c_{n})\in \mathbb{C}^{n}$,  it holds
\begin{align} \label{szhen}\begin{array}{lllllllll}\displaystyle (c_{1}, \ldots, c_{n})A_{x_{1}, \ldots, x_{n}}(c_{1}, \ldots, c_{n})^{\star}=\sum^{n}_{k=1}\sum^{n}_{l=1}c_{k}\overline{c}_{l}\varphi(x_{k}-x_{l})\geq\Lambda\sum^{n}_{j=1}|c_{j}|^{2},\end{array}\end{align}
where $(c_{1}, \ldots, c_{n})^{\star}$ is the transpose and  complex conjugate of $(c_{1}, \ldots, c_{n}).$
\end{lem}
\begin{proof}
The proof is given in section \ref{proofofuyt8}.
\end{proof}

\textbf{Proof of Proposition \ref{ertyfc}}:
Denote $P=[P_{1},   P_{2}]$. By \cite[Theorem 4.3]{Chui}, it is easy to prove that $\varphi_{2n}$ mentioned in \eqref{eeee17845}
is even, real-valued and compactly supported. Moreover,
by \eqref{Radonform}, for any $\xi\in \mathbb{R}$ we have
\begin{align} \label{fgfg} \widehat{P\phi_{B}}(\xi)=\widehat{\phi_{B}}(P^{T}\xi)=
\prod^{2}_{k=1}\Big[\frac{\sin P_{k}\xi/2}{P_{k}\xi/2}\Big]^{2n_{k}}.\end{align}
Clearly, $\widehat{P\phi_{B}}(\xi)\geq0$.  Furthermore,   $\widehat{P\phi_{B}}(\xi)>0$ whenever $|\xi|<\min\{\frac{\pi}{|p_{k}|}: k\in \hbox{supp}(P)\}$. Here  $k\in \hbox{supp}(P)$
means that $P_{k}\neq0.$ Now  by Lemma \ref{uyt8}, $P\phi_{B}$ is PD.

\subsection{Proof of Theorem  \ref{optimal}}\label{optimalll}
It follows from Proposition \ref{ertyfc} that for any $\textbf{0}\neq P\in \mathbb{R}^{1\times D}$, the projection $P\phi_{B}$
is PD. As in \eqref{XK} we have
 \begin{align}\label{XK123} f(x_{1},  x_{2})=\sum^{\lfloor b_{1}-N_{1}\rfloor}_{k_{1}=\lceil a_{1}-M_{1}\rceil} \sum^{\lfloor b_{2}-N_{2}\rfloor}_{k_{2}=\lceil a_{2}-M_{2}\rceil}c_{k_{1},  k_{2}}\phi_{B}(x_{1}-k_{1},  x_{2}-k_{2}).\end{align}
As in the proof of Theorem \ref{ee2}, define the map $P_{\mathring{E}}: \mathring{E}\longrightarrow \mathbb{R}$ by $ P_{\mathring{E}}(x)=Px, x\in \mathring{E}$.
By  \eqref{YK}  we can  denote   $Pf=\sum^{\#\Lambda}_{k=1}\widehat{c}_{j_{k}}[P\phi_{B}](\cdot-j_{k})$ where $\Lambda=\mathcal{R}(P_{\mathring{E}})=\{j_{1}, \ldots, j_{\#\Lambda}\}.$
Therefore, \begin{align}\label{eee3} \Big[Pf(j_{1}), \ldots,  Pf(j_{\#\Lambda})\Big]^{T}=A_{j_{1}, \ldots, j_{\#\Lambda}}[\widehat{c}_{j_{1}}, \ldots,
\widehat{c}_{j_{\#\Lambda}}]^{T},\end{align}
where $A_{j_{1}, \ldots, j_{\#\Lambda}}$ is defined by  \eqref{qixiang} with $\varphi$ and
$\{x_{1}, x_{2}, \ldots, x_{N}\}$ therein  replaced by $P\phi_{B}$ and $\{j_{1}, \ldots, j_{\#\Lambda}\}$, respectively.
In what follows, we prove that for a almost all   projection vector $P\in\mathcal{B}$, the map $P_{\mathring{E}}: \mathring{E} \longrightarrow\mathcal{R}(\mathring{E}) $   is a  bijection. For convenient narration denote $\mathring{E}=\{\textbf{e}_{1}, \ldots, \textbf{e}_{\#\mathring{E}}\}$.
For  any $k, j\in \{1, \ldots, \#\mathring{E}\}$, it is easy to check that  the Lebesgue measure   $\mathbf{m}\{\tilde{P}\in \mathbb{R}^{1\times 2}: \tilde{P}(\textbf{e}_{k}-\textbf{e}_{j})=0\}=0.$ Then $\mathbf{m}\Big(\bigcup^{\#\mathring{E}}_{k=1}\{\tilde{P}\in \mathbb{R}^{1\times 2}: \tilde{P}(\textbf{e}_{k}-\textbf{e}_{j})=0\}\Big)=0,$ by which it is easy to prove that
\begin{align} \label{097gr}\mathbf{m}\Big\{\tilde{P}\in\mathbb{R}^{1\times 2}:  \ \hbox{such that} \ \tilde{P}(\textbf{e}_{k}-\textbf{e}_{j})=0 \ \hbox{for} \
\hbox{some} \ k\in \{1, \ldots, \#\mathring{E}\}\Big\}=0.\end{align}
Clearly, \eqref{097gr} holds with $\tilde{P}$ replaced by $P$.
Stated another way, $P_{\mathring{E}} $   is    injective.
By Proposition \ref{ertyfc},
$P\phi$ is   PD.
Then it follows from Proposition \ref{6758} that $f$ can be determined by the SA-Radon  samples  $\{Pf(Pk): k\in \mathring{E}\}$.
The proof is concluded. $\Box$

\section{Numerical Simulation}
Let $ \phi_{B}(x_{1},  x_{2})=\varphi_{2}(x_{1})\varphi_{2}(x_{2})$ be defined by  \eqref{eeee17845}.
For   a projection vector  $P=[\cos\theta, \sin\theta]$ such that
$\tan\theta\geq2$, through direct  calculation it follows from  \eqref{en123456}  that the Radon transform
\begin{align}
 P\phi_{B}(x, \theta)=
 \left\{
 \begin{array}{lllllllll}
 \frac{(\tan\theta-\frac{x}{\cos\theta})\big[(\frac{x}{\cos\theta}-\tan\theta-\frac{3}{2})^{2}+\frac{3}{4}\big]+1}{6\cos\theta \tan^{2}\theta},\ \ & x\in(\sin\theta, \sin\theta+\cos\theta)\\
 \\
 \frac{3\big(\frac{x}{\cos\theta}-\tan\theta\big)^{2}+\big(\frac{x}{\cos\theta}-\tan\theta\big)^{3}+3\tan\theta-\frac{3x}{\cos\theta}+1}{6\cos\theta \tan^{2}\theta}, \ \ &x\in(\sin\theta-\cos\theta, \sin\theta]\\
 \\
 \frac{\tan\theta-\frac{x}{\cos\theta}}{\cos\theta \tan^{2}\theta}, \ \ &x\in[\cos\theta, \sin\theta-\cos\theta]\\
 \\
 \frac{\frac{x}{\cos\theta}\big(\frac{2x^{2}}{\cos^{2}\theta}-\frac{6x}{\cos\theta}+3\big)+6\tan\theta-\frac{3x}{\cos\theta}-2}{6\cos\theta \tan^{2}\theta}, \ \ &x\in(0, \cos\theta)\\
 \\
 \frac{6\tan\theta+\frac{6x}{\cos\theta}-2\big(\frac{x}{\cos\theta}+1\big)^{3}}{6\cos\theta \tan^{2}\theta},\ \ &x\in(-\cos\theta, 0]\\
 \\
 \frac{\tan\theta+\frac{x}{\cos\theta}}{\cos\theta \tan^{2}\theta}, \ \ &x\in[-\sin\theta+\cos\theta, -\cos\theta]\\
 \\
  \frac{3\tan\theta+\frac{3x}{\cos\theta}+1+\big(3-\tan\theta-x\big)\big( \tan\theta+\frac{x}{\cos\theta}\big)^{2}}{6\cos\theta \tan^{2}\theta},\ \ &x\in[-\sin\theta, -\sin\theta+\cos\theta)\\
  \\
  \frac{\big(\tan\theta+\frac{x}{\cos\theta}+1\big)^{3}}{6\cos\theta \tan^{2}\theta}, \ \ &x\in(-\sin\theta-\cos\theta, -\sin\theta)\\
 0,\ \ &\hbox{else}.
 \end{array}\right.
\end{align}

\begin{figure*}[htbp]\label{kncc}
 \scriptsize
 \centering
    \begin{minipage}[b]{.1\linewidth}
    \centerline{\includegraphics[width=14cm, height=5cm]{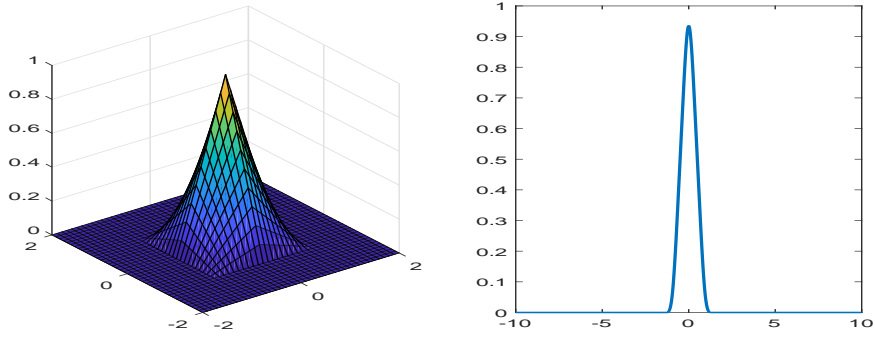}}
  \end{minipage}\hfill
 \caption{Left: the plot of $\phi_{B}$. Right: the plot of Radon transform  $  P\phi_{B}(x)$ with $P=[\cos(1.2208), \sin(1.2208)]$.}
 \label{fig:balvsunbalHam156}
\end{figure*}

\begin{figure*}[htbp]\label{kncc}
 \scriptsize
 \centering
    \begin{minipage}[b]{.1\linewidth}
    \centerline{\includegraphics[width=14cm, height=10cm]{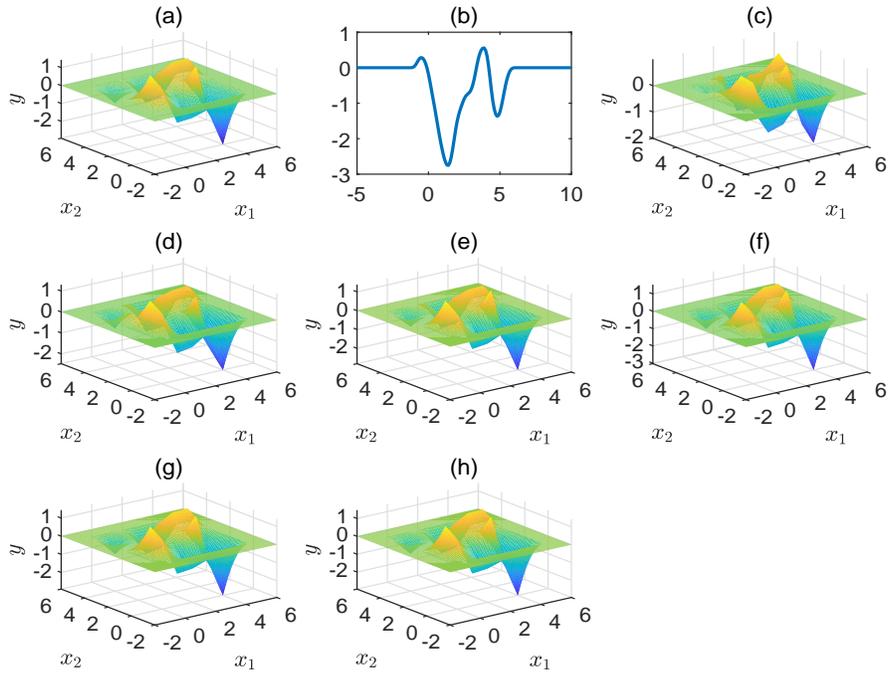}}
  \end{minipage}\hfill
 \caption{(a) The figure of $f(x_{1}, x_{2})$; (b) The figure of Radon transform $ Pf$; (c) The reconstruction version $\widetilde{f}(x_{1}, x_{2})$ when SNR=30;
 (d) The reconstruction version $\widetilde{f}(x_{1}, x_{2})$ when SNR=35; (e) The reconstruction version $\widetilde{f}(x_{1}, x_{2})$ when SNR=40;
 (f) The reconstruction version $\widetilde{f}(x_{1}, x_{2})$ when SNR=45; (g) The reconstruction version $\widetilde{f}(x_{1}, x_{2})$ when SNR=50;
 (h) The reconstruction version $\widetilde{f}(x_{1}, x_{2})$ when SNR=55.}
 \label{fig:balvsunbalHam156}
\end{figure*}

Without bias, we choose a  function
 \begin{align}\label{indexfunction} f=\sum^{4}_{k_{1}=0}\sum^{4}_{k_{2}=0}C_{k_{1}, k_{2}}\phi_{B}(x_{1}-k_{1}, x_{2}-k_{2})\in V(\phi_{B})\end{align} being supported on $[0, 6]^{2}$.
Randomly choosing an angle $\theta$ from $(\arctan 2, \pi/2)$, it follows from Theorem \ref{optimal} that with probability $1$, $f$ can be determined by its SP Radon samples $\{Pf(Pk): k\in \mathring{E}\}$
with
$\mathring{E}=[0, 4]^{2}\cap \mathbb{Z}^{2}.$  For simplicity,  by the lexicographic order  $\mathring{E}_{N}$ is denoted  as   $\{e_{1}, \ldots, e_{25}\}$. Then $f$ is determined  through the
linear equation system
\begin{align}\label{8543fr} A\big(C_{e_{1}}, \ldots, C_{e_{25}}\big)^{T}=
\big(Pf(Pe_{1}), \ldots, Pf(Pe_{25})\big)^{T},\end{align}
where $A=(P\phi_{B}(Pe_{j}-Pe_{l})\Big)^{25}_{j, l=1}.$
In this simulation we choose $\theta=1.2208$ and $P=[\cos(1.2208), \sin(1.2208)]$. The graphs of $\phi_{B}$ and $  P\phi_{B}(x)$ are  illustrated  in Figure \ref{fig:balvsunbalHam156}.
We add the  Gaussian noise $\epsilon\thicksim N(0, \sigma^{2})$ to the samples on the right-hand side of \eqref{8543fr}, and
what we observed is $Pf(Pe_{l})=Pf(Pe_{l})+\epsilon$.
Here the variance $\sigma^{2}$   is chosen such that the desired signal to noise ratio  (SNR)
is expressed by
\begin{align}\label{SNRDE} \begin{array}{lllllllllllllllll} \hbox{SNR}=10\log_{10}\Big(\frac{\|\{Pf(Pe_{l}): \
 l=1, \ldots, 25\}\|^{2}_{2}}{25\sigma^{2}}\Big).\end{array}\end{align}
We consider the
Tikhonov regularization problem
\begin{align}\label{uygr}
\arg\min_{C}\|AC-\big(\widetilde{Pf(Pe_{1})}, \ldots, \widetilde{Pf(Pe_{25})}\big)^{T}\|_{2}
+\alpha\|C\|_{2},
\end{align}
where the regularization parameter $\alpha$ satisfies the requirement in \cite[Theorem 2.5]{Engl}.
In this simulation, $\alpha$ is chosen as in Table 1.
The solution to \eqref{uygr} is \begin{align}\label{jie} \widehat{C}=(\alpha I+A^{*}A)^{-1}A^{*}\big(\widetilde{Pf(Pe_{1})}, \ldots, \widetilde{Pf(Pe_{25})}\big)^{T},\end{align}
where $I$ is the identity matrix.
For every $\hbox{SNR}\in \{30, 35, 40, 45, 50, 55\}$, conduct \eqref{jie} for  $1000$ trials
and the average of the $1000$ outputs   is denoted by $\bar{\widehat{C}}$. Construct $\widetilde{f}$ through
\eqref{indexfunction} with $(C_{e_{1}}, \ldots, C_{e_{25}}\big)^{T}$ being replaced by
$\bar{\widehat{C}}$.
And  the  reconstruction error   is defined as
\begin{align}
\hbox{error}=
\frac{\Big[\sum^{80}_{k=0}\sum^{80}_{l=0}\big(\widetilde{f}(-2+0.1k, -2+0.1l)-f(-2+0.1k, -2+0.1l)\big)^{2}\Big]^{1/2}}{\Big[\sum^{80}_{k=0}\sum^{80}_{l=0}\big(f(-2+0.1k, -2+0.1l)\big)^{2}\Big]^{1/2}}.
\end{align}
The average errors corresponding to different choices of SNR are  recorded in Table 1.

\begin{table*}[htbp]\tiny
\label{Table0}
\begin{center}
\begin{tabular}{|c|c|c|c|c|c|c|c}
  \hline
  SNR & 30 & 35 & 40  & 45 & 50 & 55  \\ \hline
  $\alpha$ &$2e-03$&$2.5e-04$&$3.5e-05$&$2e-06$&$1.5e-07$&$9.4e-08$                                              \\ \hline
  \hbox{error} & 0.4255 & 0.2471 & 0.0924  & 0.0480 & 0.0290 & 0.0087 \\ \hline
\end{tabular}
 \caption{ The error  vs  SNR.}
\end{center}
\label{opqerr}
\end{table*}

\section{Conclusion}
We establish the single-angle (SA)   Radon samples based reconstruction for the compactly supported functions
in refinable shift-invariant spaces (SIS).
We  prove that  any compactly supported function in   a  general refinable SIS
can be reconstructed by its  Radon samples at an  appropriate   angle. We also
investigate the SA reconstruction problem in a class of   SISs generated by  a class of refinable
box-splines. It is proved that  any compactly supported function in the SISs generated by such box-splines,
can be reconstructed by the SA Radon samples at almost every angle in $[0, 2\pi)$.

\section{Appendix: Proof of Lemma \ref{uyt8}}\label{proofofuyt8}
Since $\varphi$ is even and real-valued,  so is $\widehat{\varphi}$.
Consequently, $A_{x_{1}, \ldots, x_{n}}$ is symmetric.
Now we choose  $0<\tilde{\zeta}\leq\zeta$ such that $\widehat{\varphi}(\xi)>0$  on $[-\tilde{\zeta}, \tilde{\zeta}]$.
Then for any   $\textbf{0}\neq(c_{1}, \ldots, c_{n})\in \mathbb{C}^{n}$, it follows from Lemma \ref{uyt7}    that
\begin{align}\label{eeee890} \begin{array}{lllllllll} \displaystyle\sum^{n}_{k=1}\sum^{n}_{l=1}c_{k}\overline{c}_{l}\varphi(x_{k}-x_{l})&=(c_{1}, \ldots, c_{n})A_{x_{1}, \ldots, x_{n}}(c_{1}, \ldots, c_{n})^{\star}\\
&\displaystyle=\int_{\mathbb{R}}\widehat{\varphi}(\xi)|\sum^{n}_{k=1}
c_{k}e^{-\textbf{i}x_{k}\xi}|^{2}d\xi\\
&\geq \displaystyle\int^{\tilde{\zeta}/2}_{-\tilde{\zeta}/2}\widehat{\varphi}(\xi)|\sum^{n}_{k=1}
c_{k}e^{-\textbf{i}x_{k}\xi}|^{2}d\xi,
\end{array}
\end{align}
where in the inequality  we use $\widehat{\varphi}(\xi)\geq0$ for any $\xi\in \mathbb{R}$.
Since  $x_{k}\neq x_{j}$,
it follows from \cite[Lemma 6.7]{Wendlan}  that $\{e^{-\textbf{i}x_{k}\xi}\}^{n}_{k=1}$
is linear independent in the space $L^{2}([-\tilde{\zeta}/2, \tilde{\zeta}/2])$.
Consequently,   $\{e^{-\textbf{i}x_{k}\xi}\}^{n}_{k=1}$ is a Riesz basis  for its finitely spanning  space
$\{\sum^{n}_{k=1}c_{k}e^{-\textbf{i}x_{k}\xi}: c_{k}\in \mathbb{C}\}\subseteq L^{2}([-\tilde{\zeta}/2, \tilde{\zeta}/2])$, and there exists $\lambda_{X, \tilde{\zeta}}>0$ (depending on $X$
and $\tilde{\zeta}$)
such that
\begin{align}\label{ouyang} \int^{\tilde{\zeta}/2}_{-\tilde{\zeta}/2}|\sum^{n}_{k=1}
c_{k}e^{-\textbf{i}x_{k}\xi}|^{2}d\xi\geq \lambda_{X, \tilde{\zeta}} \|\{c_{k}\}^{n}_{k=1}\|^{2}_{2}.\end{align}
On the other hand, since $\varphi\in C(\mathbb{R})$ is  compactly supported,
then $\widehat{\varphi}\in C^{\infty}(\mathbb{R})$. Therefore, $\min\{\widehat{\varphi}(\xi): \xi\in [-\zeta, \zeta]\}$
exists and is positive. Consequently, combining   \eqref{eeee890} and \eqref{ouyang} we have
\begin{align}\label{eeee1000} \begin{array}{lllllllll} \displaystyle\sum^{n}_{k=1}\sum^{n}_{l=1}c_{k}\overline{c}_{l}\varphi(x_{k}-x_{l})&=(c_{1}, \ldots, c_{n})A_{x_{1}, \ldots, x_{n}}(c_{1}, \ldots, c_{n})^{\star}\\
&\geq \lambda_{X, \tilde{\zeta}} \|\{c_{k}\}^{n}_{k=1}\|^{2}_{2}>0.
\end{array}
\end{align}
Stated another way, the matrix  $A_{x_{1}, \ldots, x_{n}}$ is positive definite.
Therefore, by \eqref{eeee1000} we have
\begin{align}\begin{array}{lllllllll}
\displaystyle\sum^{n}_{k=1}\sum^{n}_{l=1}c_{k}\overline{c}_{l}\varphi(x_{k}-x_{l})
&\geq\displaystyle\sup_{0<\tilde{\zeta}\leq\zeta}\Big\{\lambda_{X, \tilde{\zeta}}\min\{\widehat{\varphi}(\xi): \xi\in U(0, \tilde{\zeta})\}\Big\}\sum^{n}_{j=1}|c_{j}|^{2}.
\end{array}\end{align}
Now we choose $0<\Lambda:=\sup_{\tilde{\zeta}}\Big\{\lambda_{X, \tilde{\zeta}}\min\{\widehat{\varphi}(\xi): \xi\in U(0, \tilde{\zeta})\}\Big\}<\infty$
to conclude the proof.

\end{document}